%% file: mendoza.tex
\documentclass[reprint,aps,showkeys,nofootinbib]{revtex4-2}

\usepackage{graphicx}

\usepackage{hyperref}

\usepackage{amsmath}
\usepackage{amssymb}
\usepackage{amsfonts}
\usepackage{amsthm} 

\usepackage{caption}

\usepackage[british]{babel}

\usepackage{color}

\usepackage{ulem}

\usepackage{cancel}

\usepackage[newitem,newenum,defblank]{paralist}

\input{aas-journals}

\newcommand{\Ocal}{\mathcal{O}}

\newtheorem{theorem}{Theorem}
\newtheorem{corollary}{Corollary}[theorem]

\begin{document}

\title{Metric tensor at second perturbation order for
spherically symmetric space-times}


\author{Sergio Mendoza}
\email{sergio@astro.unam.mx}
\affiliation{Instituto de Astronom\'{\i}a, Universidad Nacional Aut\'onoma de M\'exico, AP 70-264, Ciudad de M\'exico 04510, M\'exico}

\date{\today}


\begin{abstract}
  It is shown in this article that if the Einstein Equivalence Principle is
valid on a particular metric theory of gravitation in a spherically symmetric
space-time, then the time metric component is not equal to the negative of
the inverse radial one unless the underlying potential is inversely 
proportional to the radial coordinate.  At the weak field limit of
approximation, a general formula is calculated and applied to
some useful cases.
\end{abstract}


\maketitle

\section{Introduction}
\label{introduction}

  Finding solutions to extended metric theories of gravity is a difficult
task, even when simple static and spherical approximations are assumed.
A simple general technique that is often used in the literature to find
solutions in spherical space-times is in which the time component of
the metric equals the negative of the inverse radial one.  This fact is 
motivated by the
validity of this statement in the Schwarzschild solution of general
relativity.  However, the procedure to build Schwarzschild's solution does
not start with such assumption. It follows as a result of the general
solution.

  In this article it is shown that if a metric theory of gravity 
obeys the Einstein Equivalence Principle, then the metric time component is
different from the negative inverse radial one.  By adopting a Post
Newtonian coordinate system of reference and a Post Newtonian gauge it is
obtained the correct value of those metric components at the weak field limit
of approximation in Schwarzschild-like coordinates.  

  The article is organised as follows.  Section~\ref{isotropicspace}
discusses in general terms the Parametrised Post Newtonian (PPN) formalism
at second order of approximation (the weak field limit) in standard Post
Newtonian coordinates -which are equivalent to an isotropic space, and in 
Schwarzschild-like spherical ones.    In Section~\ref{sphericalspace}, the
transformation from the isotropic space to the spherical one is performed
and the metric components are calculated at second order of approximation.
Useful examples of the general result are presented in
Section~\ref{applications} and finally we discuss our results in
Section~\ref{discussion}.

\section{Isotropic space-time.}
\label{isotropicspace}

  The geodesic motion of a massive particle with mass \( m \) is obtained by
minimising the action \citep[see e.g.][]{landau-fields}:

\begin{equation}
   S = - m c \int_a^b{ \mathrm{d} s } = \int_a^b{ \mathcal{L} \mathrm{d}t },
\label{geodesic}
\end{equation}

\noindent with respect to the space-time coordinates \( x^\alpha
\).   In the previous equation, \( \mathrm{d} s = g_{\alpha\beta}
\mathrm{d}x^\alpha \mathrm{d} x^\beta \) represents the interval of space
time for a metric \( g_{\alpha\beta} \) and \( x^\alpha\) are a set of
chosen coordinates, with \( x^0=t \) representing the time coordinate
and \( x^k \) spatial coordinates.  In here and in what follows we use
Einstein's summation convention and Greek space-time indices run from \(
0 \) to \( 4 \) while Latin space ones from \( 1 \) to \( 3 \).  We choose
a metric signature (\( +,-,-,- \)). At the weakest order of approximation,
i.e. when the velocity of light \( c \rightarrow \infty \), the Lagrangian
\( \mathcal{L} \) is given by~\citep[see e.g.][]{landau-fields,willbook}:

\begin{equation}
  \mathcal{L} = - m c^2 + \frac{ 1 }{ 2 } m v^2  - m \phi,
\label{lagrangian}
\end{equation}

\noindent where \( \phi \) represents the Newtonian scalar potential,
which for the particular case of a point mass source \( M \) at the
origin of coordinates is given by:

\begin{equation}
  \phi = - \frac{ G M }{ r },
\label{newtonian-potential}
\end{equation}

\noindent in which \( G \) is Newton's gravitational constant and \( r \)
the radial distance to the origin. Combining equations~\eqref{lagrangian}
and~\eqref{geodesic}, the time component \( g_{00} \) of the metric \(
g_{\alpha\beta} \) to \( \Ocal(1/c^2) \) is given by:

\begin{equation}
  g_{00} = {}^{(0)}g_{00} + {}^{(2)}g_{00} = 1 + \frac{ 2 \phi }{ c^2 }.
\label{weakest}
\end{equation}

  In the previous equation and in what follows we will
refer to perturbation orders as \( \Ocal(2) \), \( \Ocal(4)
\) meaning \( \Ocal(1/c^2) \), \( \Ocal(1/c^4) \) and so on.
In equation~\eqref{weakest} and in what follows, the left superindex
parenthesis in a particular quantity represents the perturbation order
at which that particular quantity is approximated.
The fact that \( {}^{(2)}g_{00} = 2 \phi / c^2 \) is a manifestation of
the validity of the Einstein Equivalence Principle~\citep{willbook}.

  In a Post-Newtonian system of reference with a standard Post-Newtonian
gauge at \( \mathcal{O}(2) \) of approximation, the metric is isotropic
and of the form~\citep{willbook}:

%
%

\begin{equation}
  \mathrm{d}s^2 = g_{00} c^2 \mathrm{d}t^2 + \Lambda(\boldsymbol{x})
     \mathrm{d}l^2,
\label{isotropic-tmp}
\end{equation}

\noindent where the square of the three dimensional length element 
\( \mathrm{d}l^2 = \mathrm{d} x_k \mathrm{d} x^k \) and \( \boldsymbol{x} \) 
is a three dimensional coordinate vector. In spherical-like
coordinates (\(t,\tilde{r},\theta,\varphi\)) the previous equation is:

\begin{equation}
  \mathrm{d}s^2 = g_{00} c^2 \mathrm{d}t^2 + \Lambda(\boldsymbol{x}) \left(
  \mathrm{d} \tilde{r} + \tilde{r}^2 \mathrm{d}\Omega \right),
\label{isotropic}
\end{equation}

\noindent where the squared angular displacement \( \mathrm{d} \Omega^2 :=
\mathrm{d} \theta^2 + \sin^2 \theta \mathrm{d} \varphi^2 \) for the polar and
azimuthal angles \( \theta \) and \( \varphi \) respectively.  The radial
distance is represented by \( \tilde{r} \).

  At \( \Ocal(0) \) the components of the metric \( g_{00} = 1 \) and \(
\Lambda = - 1 \).  The next \( \Ocal(2) \) correction for the time
component of the metric is given by equation~\eqref{weakest}.  To find the
\( \Ocal(2) \) correction to the function \( \Lambda(\boldsymbol{x}) \)
note that so far we have not introduced any other function into the
discussion but the function \( \phi(r) \).  As such, it is expected that 
\( {}^{(2)}\Lambda \propto \phi \) and so the line
element~\eqref{isotropic} can be rewritten as:

\begin{equation}
  \mathrm{d}s^2 = \left( 1 + \frac{ 2 \phi(\tilde{r})}{ c^2 } \right)
    c^2 \mathrm{d}t^2 - \left( 1 - \frac{ 2 \gamma \phi(\tilde{r}) }{
    c^2 } \right) \left(
  \mathrm{d} \tilde{r} + \tilde{r}^2 \mathrm{d}\Omega \right),
\label{isotropic-final}
\end{equation}

\noindent Notice that the function \( {}^{(2)}\Lambda \) has been written
to resemble as much as possible the function \( {}^{(2)}g_{00} \) and that
is the reason for introducing the constant parameter \( \gamma \) which
is the first PPN parameter as discussed in Section~\ref{introduction}.

  The PPN parameter \( \gamma \) measures the amount of curvature of space
\citep{willbook} and its precise value must be obtained experimentally.
A viable theory of gravity must converge to that value at its \(
\mathcal{O}(2) \) order of approximation.  The deflection of light
observed in Solar System experiments yields a value \( \gamma = 1 \). This
value is fully predicted by general relativity at \( \mathcal{O}(2) \)
perturbation order of the theory.  For the case of MOdified theories
of gravity (MOND), it turns out that the same value is also obtained
in observations of the deflection of light in individual, groups and
clusters of galaxies~\citep{mendoza13,mendozaolmo,mendoza15}.

\section{Spherical symmetry}
\label{sphericalspace}

  Let us make a coordinate transformation to spherical coordinates 
(\(t, r, \theta, \varphi\)) so that
the line element~\eqref{isotropic-final} is given by:

\begin{equation}
  \mathrm{d} s^2 = \left( 1 + \frac{ 2 \phi(r) }{ c^2 } \right) c^2
  \mathrm{d}t^2 - g_{11} \mathrm{d}r^2 - r^2 \mathrm{d} \Omega^2.
\label{spherical}
\end{equation}

\noindent In the previous equation we have left the time component of the
metric \( g_{00} = \left( 1 + 2 \phi(r) / c^2 \right) \) unchanged in order to
comply with the lowest perturbation order obtained from the
action~\ref{geodesic}, which is coherent to the Einstein
Equivalence Principle.  The angular displacement \( \mathrm{d} \Omega \)
has been left unchanged with the transformation since the coordinates \(
\theta \) and \( \varphi \) are physically the same between both systems of
reference.

  From equation~\eqref{spherical} and~\eqref{isotropic-final} it follows
that:

\begin{gather}
  r = \left( 1 - \frac{ \gamma \phi(\tilde{r}) }{c^2} \right)  \tilde{r},
  			\label{t01} \\
  g_{11} \mathrm{d}r^2 = - \left( 1 - \frac{ 2 \gamma \phi(\tilde{r}) }{c^2} 
     \right) \mathrm{d}\tilde{r}^2,
     			\label{t02}
\end{gather}

\noindent at \( \Ocal{(2)} \), and so, relation~\eqref{t02} is:

\begin{equation}
  \mathrm{d}r = \left( 1 - \frac{ \gamma \phi(\tilde{r}) }{ c^2 } \right) 
     \mathrm{d}\tilde{r} - \frac{ \gamma \tilde{r} \mathrm{d}
     \phi(\tilde{r}) }{ c^2 }.
\label{diff}
\end{equation}

\noindent Substitution of this previous result into
equation~\eqref{t02} yields at \( \Ocal(2) \):

\begin{equation}
  g_{11}(r) = - \left( 1 + \frac{ 2 \gamma \tilde{r} }{ c^2 } \frac{
  \mathrm{d} \phi(\tilde{r}) }{ \mathrm{d} \tilde{r} }  \right),
\label{g11}
\end{equation}

\noindent where we have used the fact that \( \mathrm{d} \phi(\tilde{r}) 
= \mathrm{d} \tilde{r} \, \mathrm{d} \phi(\tilde{r}) / \mathrm{d} \tilde{r}
\). The function \( g_{11}(r) \) in the previous equation is only a
function of the spherical radial coordinate \( r \) according to
equation~\eqref{t01}.  To see this in more detail, note that:

\begin{displaymath}
  \phi( \tilde{r} ) = \phi \left( r + \frac{ \gamma r \phi(\tilde{r})
    }{ c^2 } \right),
\end{displaymath}

\noindent and so, performing a Taylor expansion of the previous function up to
\( \Ocal(2) \) it follows that:

\begin{equation}
  \phi( \tilde{r} ) = \phi( r ) + \frac{ \gamma r \phi(\tilde{r}) 
    }{ c^2 } \frac{ \mathrm{d} \phi(r) }{ \mathrm{d} r }.
\end{equation}

\noindent Direct substitution of this relation into equation~\eqref{g11}
and using relations~\eqref{t01} and~\eqref{diff} it follows that:

\begin{equation}
 \begin{split}
  \frac{ \mathrm{d} \phi(\tilde{r}) }{ \mathrm{d} \tilde{r} } &= 
    \frac{ \mathrm{d} \phi(r) }{ \mathrm{d} \tilde{r} } + 
    \frac{ \mathrm{d} }{ \mathrm{d} \tilde{r} } \left[ \frac{ \gamma r
    \phi( \tilde{r} ) }{ c^2 } \right] 
    					\\
  &= \left[ 1 - \frac{ \gamma \phi(\tilde{r}) }{ c^2 } - \frac{ \gamma
    \tilde{r} }{ c^2 } \frac{ \mathrm{d} \phi( \tilde{r} ) }{ \mathrm{d}
    \tilde{r} } \right]  
    \left\{ \frac{ \mathrm{d} \phi(r) }{ \mathrm{d} \tilde{r} } + 
      \frac{ \mathrm{d} }{ \mathrm{d} \tilde{r} } \left[ \frac{ \gamma r
      \phi( \tilde{r} ) }{ c^2 } \right] 
    \right\},
  \end{split}
\end{equation}

\noindent which at \( \Ocal(2) \) perturbation order yields:

\begin{displaymath}
  \frac{ \mathrm{d} \phi( \tilde{r} ) }{ \mathrm{d} \tilde{r} } \left(
    1 + \frac{ \gamma r }{ c^2 } \frac{ \mathrm{d} \phi(\tilde{r}) }{ 
    \mathrm{d} r } \right) = \frac{ \mathrm{d} \phi(r) }{ \mathrm{d} r} +
    \frac{ \gamma r \phi }{ c^2 } \frac{ \mathrm{d}^2 \phi(r) }{ \mathrm{d}
    r^2 },
\end{displaymath}

\noindent and so:

\begin{equation}
  \frac{ \mathrm{d} \phi( \tilde{r} ) }{ \mathrm{d} \tilde{r} } =
  \frac{ \mathrm{d} \phi(r) }{ \mathrm{d} r } + \frac{ \gamma r }{ c^2 }
  \left\{ - \left( \frac{ \mathrm{d} \phi(r) }{ \mathrm{d} r } \right)^2 
  + \phi \frac{ \mathrm{d}^2 \phi(r) }{ \mathrm{d} r^2 } \right\}.
\label{der}
\end{equation}

  The substitution of this last result into equation~\eqref{g11} at the
same \( \Ocal(2) \) perturbation order  yields:

\begin{equation}
  g_{11}(r)  =  - \left( 1 + \frac{ 2 \gamma r }{ c^2 } \frac{ \mathrm{d}
    \phi(r) }{ \mathrm{d} r }    \right) 
\end{equation}

\noindent The previous results can be stated formally in the following way:

\begin{theorem}
  The general form of the metric in spherical coordinates for 
any metric theory of gravity at \( \Ocal(2) \) perturbation order can be
written as:

\begin{equation}
  \mathrm{d} s^2 = \left( 1 + \frac{ 2 \phi(r) }{ c^2 } \right) c^2
    \mathrm{d}t^2 - \left( 1 + \frac{ 2 \gamma r }{ c^2 } \frac{ \mathrm{d}
    \phi(r) }{ \mathrm{d} r } \right) \mathrm{d}r^2 - r^2 \mathrm{d}
    \Omega^2,
\label{metric-final01}
\end{equation}

\noindent or in Schwarzschild-like space-time form:

\begin{equation}    				\\
  \mathrm{d} s^2 = \left( 1 + \frac{ 2 \phi(r) }{ c^2 } \right) c^2 
    \mathrm{d}t^2
    - \frac{ \mathrm{d} r^2 }{ \left( 1 - \frac{ 2 \gamma r }{ 
    c^2 } \frac{ \mathrm{d} \phi(r) }{ \mathrm{d} r } \right) }- 
    r^2 \mathrm{d} \Omega^2.
\label{metric-final02}
\end{equation}

\label{theorem01}
\end{theorem}

\begin{corollary}
  The only potential that satisfies a Schwarzschild-like behaviour for 
which \( g_{00} = - 1 / g_{11}  \) is one that satisfies \( \phi \propto
1 / r^{1/\gamma} \).  
\label{corollary01}
\end{corollary}

\begin{proof}
  The condition \( g_{00} = - 1 / g_{11}  \) applied to the
interval~\eqref{metric-final02}  yields \( \gamma \mathrm{d} \ln \phi   = -
\mathrm{d} \ln r  \), i.e. \( \phi \propto 1 / r^{1/\gamma} \).
\end{proof}

  As mentioned at the end of Section~\ref{isotropic}, different Solar
System experiments show that \( \gamma = 1 \) and
so the potential \( \phi \propto 1 / r \) is the only one that satisfies
the conditions of the previous Corollary.   This corollary is fully
satisfied in general relativity, but as we will see in the next Section, it
is not satisfied in general metric theories of gravity.

\section{Applications}
\label{applications}

  In this section we discuss some of the applications which are  important
for different cases in many extended theories of gravity which do not
necessarily have a Newtonian potential weak-field limit of approximation.

\subsection{Point particle Newtonian-like potentials.}
\label{newtonian}

  Let us assume that:

\begin{equation}
  \phi(r) = A r^k, \qquad \text{with} \qquad k \neq 0,
\label{ntw}
\end{equation}

\noindent so that:

\begin{displaymath}
  \frac{ \mathrm{d} \phi(r) }{ \mathrm{d} r } = A k r^{k-1},
\end{displaymath}

\noindent and so, from the results of Theorem~\ref{theorem01}
it follows that:

\begin{equation}
  \begin{split}
  \mathrm{d}s^2 &= \left( 1 + \frac{ 2 A r^k }{ c^2 } \right) c^2
    \mathrm{d}t^2 + \left( 1 + \frac{ 2 \gamma A k r^k }{c^2} \right)
    \mathrm{d} r^2 + r^2 \mathrm{d} \Omega^2,
     		 \\
     &= \left( 1 + \frac{ 2 A r^k }{ c^2 } \right) c^2
    \mathrm{d}t^2 + \frac{ \mathrm{d} r^2 }{ \left( 1 -
      \frac{ 2 \gamma A k r^k }{c^2} \right) } + r^2 \mathrm{d} \Omega^2
  \end{split}
\label{g11nt}
\end{equation}

\noindent at perturbation order \( \Ocal(2) \).

  For the case of a Newtonian potential produced by a point mass particle
located at the origin of coordinates \( k = -1 \) and \( A = -G M \), where
\( G \) is Newton's constant of gravitation and \( M \) is the mass of the
particle producing the gravitational field.  With this and using the
results of equation~\eqref{g11nt} and~\eqref{spherical} it follows that:

\begin{equation}
  \begin{split}
  \mathrm{d} s^2 &= \left( 1 - \frac{ 2 G M }{ r c^2 } \right) c^2
    \mathrm{d}t^2 - \left( 1  + \frac{ 2 \gamma G M }{ r c^2 } \right)
    \mathrm{d}r^2 - r^2 \mathrm{d} \Omega^2.
    		\\
    &= \left( 1 - \frac{ 2 G M }{ r c^2 } \right) c^2
    \mathrm{d}t^2 - \frac{ \mathrm{d}r^2 }{ \left( 1  - 
    \frac{ 2 \gamma G M }{ r c^2 } \right)} - r^2 \mathrm{d} \Omega^2.
  \end{split}
\label{schwarzschild}
\end{equation}

\noindent This reproduces Schwarzschild's metric at \( \Ocal(2) \) of
approximation if \( \gamma = 1 \), which as mentioned before is fulfilled
by a wide number of experiments.

\subsection{Logarithmic potential}
\label{logarithmic}

  Logarithmic potentials appear on MOdified Newtonian Dynamics (MOND)
gravitation~\citep[see e.g.][and references therein]{mendoza15}.
The deep MOND regime is obtained when the acceleration exerted on a
test particle is given by:

\begin{equation}
  a = \frac{ \mathrm{d} \phi(r)  }{ \mathrm{d} r} = 
    - \frac{ \sqrt{G M a_0} }{ r }, 
\label{mond-acc}
\end{equation}

\noindent and so, the potential is given by:

\begin{equation}
  \phi(r) = - \sqrt{ G M a_0 } \ln( r /r_\star ),
\label{mond-pot}
\end{equation}

\noindent where \( r_\star \) is an arbitrary distance and \( a_0 \approx
1.2 \times 10^{-10} \textrm{m} / \textrm{s}^2 \) is Milgrom's acceleration
constant.  Direct substitution of the previous two equations into
equation~\eqref{g11} yields at perturbation order \( \Ocal(2) \):

\begin{equation}
  \begin{split}
  \mathrm{d} s^2 &= \left( 1 - \frac{ 2 \sqrt{ G M a_0 } }{ c^2 } 
    \, \ln( r / r_\star  ) \right) c^2 \mathrm{d}t^2 -
    				\\
    & \left( 1 - \frac{ 2 \gamma }{ c^2 } \sqrt{ G M a_0 } \right)
    \mathrm{d}r^2 - r^2 \mathrm{d} \Omega^2,
    				\\
    &= \left( 1 - \frac{ 2 \sqrt{  G M a_0 } }{ c^2 } 
    \, \ln( r / r_\star  ) \right) c^2 \mathrm{d}t^2 -
    				\\
    & \frac{ \mathrm{d} r^2 }{ \left( 1 + \frac{ 2 \gamma }{ c^2 } 
      \sqrt{ G M a_0 } \right) } - r^2 \mathrm{d} \Omega^2.
  \end{split}
\label{metric-mond}
\end{equation}

\noindent The previous metric was obtained by \citet{mendozaolmo} and
it was verified by \citet{mendoza13} that the parameter \( \gamma =
1 \) from observations of the bending of light of individual, groups
and clusters of galaxies.

\subsection{Yukawa-like potential}
\label{yukawa-like}

  For the case of analytic \( f(R) \) theories of gravity it has been shown
by \citet{capozziello12} that:

\begin{equation}
  \phi(r) = - \frac{ G M }{r } \frac{ 1 + \delta e^{r/\lambda} }{ \delta +
    1 }, 
\label{yukawa-pot}
\end{equation}

\noindent and so:

\begin{equation}
  \frac{ \mathrm{d} \phi(r) }{ \mathrm{d} r}  = \frac{ G M }{ r \left(
   \delta + 1 \right)  } \left\{ \frac{ 1 }{ r } + \delta e^{r/\lambda} 
   \left[ \frac{ 1 }{ r } + \frac{ 1 }{ \lambda} \right] \right\}.
\label{yukawa-der}
\end{equation}

\noindent Substitution of these last two expressions on the results of
Theorem~\ref{theorem01} yields:

\begin{equation}
\begin{split}
 \mathrm{d}s^2 &=  \left\{ 1 - \frac{ 2 G M }{ r c^2 
   \left( \delta + 1 \right) } \left(  1 + \delta e^{r/\lambda}
   \right) \right\} c^2 \mathrm{d} t^2 -
   			\\
   &\left\{ 1 + \frac{ 2 \gamma G M }{  r c^2 \left(
   \delta + 1 \right)  } \left( 1 + \delta e^{-r/\lambda} \right)
   + \frac{ 2 \gamma GM  }{ \lambda \left( 1 + 
   \delta \right) } \delta e^{-r/\lambda} \right\} \mathrm{d} r^2  -
   			\\
   &r^2 \mathrm{d} \Omega^2,
   			\\
   &=  \left\{ 1 - \frac{ 2 G M }{ r c^2 
   \left( \delta + 1 \right) } \left(  1 + \delta e^{r/\lambda}
   \right) \right\} c^2 \mathrm{d} t^2 -
   			\\
   & \frac{ \mathrm{d} r^2 }{ \left\{ 1 - \frac{ 2 \gamma G M }{  r c^2 \left(
   \delta + 1 \right)  } \left( 1 + \delta e^{-r/\lambda} \right)
   - \frac{ 2 \gamma GM  }{ \lambda \left( 1 + 
   \delta \right) } \delta e^{-r/\lambda} \right\} }  -
   			\\
   &r^2 \mathrm{d} \Omega^2.
\end{split}
\label{yuka-final01}
\end{equation}

\noindent or equivalently:

\begin{equation}
\begin{split}
 \mathrm{d} &s^2 =  \left\{ 1 - \frac{ 2 G M }{ r c^2 
   \left( \delta + 1 \right) } \left(  1 + \delta e^{r/\lambda}
   \right) \right\} c^2 \mathrm{d} t^2 -
   			\\
   & \frac{ \mathrm{d} r^2 }{ \left\{ 1 - \frac{ 2 \gamma G M }{  c^2 \left(
   \delta + 1 \right)  } \left[ \frac{ 1 }{ r } + \delta e^{r/\lambda} 
   \left( \frac{ 1 }{ r } + \frac{ 1 }{ \lambda} \right) \right]
  \right\} } + r^2 \mathrm{d} \Omega^2.
\end{split}
\label{yuka-final02}
\end{equation}

  The radial component of the metric in the previous equation does not
converge to the results presented by
\citet{capozziello07,capozziello11,capozziello12,demartino14,demartino18,delaurentis18,cruz-osorio21}.
This fact occurs because in these works authors assumed that the
time and radial component of the metric satisfy \( g_{00} = - 1/ g_{11} \)
which according to the results of Corollary 1.1 is only valid for
potentials \( \phi(r) \propto 1 / r \) at \( \mathcal{O}(2) \) of
approximation for \( \gamma = 1 \).

\section{Discussion}
\label{discussion}

  The intention of this article has been to show that if the Einstein
Equivalence Principle is to be valid in a spherically symmetric space-time, 
then the interval:

\begin{equation}
  \mathrm{d} s^2 \neq \left( 1 + \frac{ 2 \phi(r) }{ c^2 } \right) c^2 
    \mathrm{d}t^2
    - \frac{ \mathrm{d} r^2 }{ \left( 1  + \frac{ 2 \phi(r) }{ c^2 }
      \right) } - r^2 \mathrm{d} \Omega^2,
\label{wrong-metric}
\end{equation}

\noindent for any metric theory of gravity in spherical Schwarzschild-like
coordinates (\(t,r,\theta,\varphi\)).  This fact follows from the remarks
of Theorem~\ref{theorem01} and  Corollary~\ref{corollary01} that were
calculated at \( \mathcal{O}(2) \) of approximation, and hence imply the
general result of equation~\eqref{wrong-metric} at any other order of
approximation.

\section*{Acknowledgements} 
This work was supported by  PAPIIT DGAPA-UNAM (IN110522) and 
CONACyT (26344).

\bibliographystyle{apsrev4-2}
\bibliography{mendoza}

\end{document}

%% file: aas-journals.tex
\newcommand{\apss}{Astrophysics and Space Science}

\newcommand{\mnras}{MNRAS}

\newcommand{\physrep}{Physics Reports}

%% file: mendoza.bbl
\begin{thebibliography}{12}%
\makeatletter
\providecommand \@ifxundefined [1]{%
 \@ifx{#1\undefined}
}%
\providecommand \@ifnum [1]{%
 \ifnum #1\expandafter \@firstoftwo
 \else \expandafter \@secondoftwo
 \fi
}%
\providecommand \@ifx [1]{%
 \ifx #1\expandafter \@firstoftwo
 \else \expandafter \@secondoftwo
 \fi
}%
\providecommand \natexlab [1]{#1}%
\providecommand \enquote  [1]{``#1''}%
\providecommand \bibnamefont  [1]{#1}%
\providecommand \bibfnamefont [1]{#1}%
\providecommand \citenamefont [1]{#1}%
\providecommand \href@noop [0]{\@secondoftwo}%
\providecommand \href [0]{\begingroup \@sanitize@url \@href}%
\providecommand \@href[1]{\@@startlink{#1}\@@href}%
\providecommand \@@href[1]{\endgroup#1\@@endlink}%
\providecommand \@sanitize@url [0]{\catcode `\\12\catcode `\$12\catcode
  `\&12\catcode `\#12\catcode `\^12\catcode `\_12\catcode `\%12\relax}%
\providecommand \@@startlink[1]{}%
\providecommand \@@endlink[0]{}%
\providecommand \url  [0]{\begingroup\@sanitize@url \@url }%
\providecommand \@url [1]{\endgroup\@href {#1}{\urlprefix }}%
\providecommand \urlprefix  [0]{URL }%
\providecommand \Eprint [0]{\href }%
\providecommand \doibase [0]{https://doi.org/}%
\providecommand \selectlanguage [0]{\@gobble}%
\providecommand \bibinfo  [0]{\@secondoftwo}%
\providecommand \bibfield  [0]{\@secondoftwo}%
\providecommand \translation [1]{[#1]}%
\providecommand \BibitemOpen [0]{}%
\providecommand \bibitemStop [0]{}%
\providecommand \bibitemNoStop [0]{.\EOS\space}%
\providecommand \EOS [0]{\spacefactor3000\relax}%
\providecommand \BibitemShut  [1]{\csname bibitem#1\endcsname}%
\let\auto@bib@innerbib\@empty
\bibitem [{\citenamefont {Landau}(2013)}]{landau-fields}%
  \BibitemOpen
  \bibfield  {author} {\bibinfo {author} {\bibfnamefont {L.}~\bibnamefont
  {Landau}},\ }\href@noop {} {\emph {\bibinfo {title} {The Classical Theory of
  Fields}}},\ Course of Theoretical Physics\ (\bibinfo  {publisher} {Elsevier
  Science},\ \bibinfo {year} {2013})\BibitemShut {NoStop}%
\bibitem [{\citenamefont {Will}(1993)}]{willbook}%
  \BibitemOpen
  \bibfield  {author} {\bibinfo {author} {\bibfnamefont {C.}~\bibnamefont
  {Will}},\ }\href@noop {} {\emph {\bibinfo {title} {Theory and Experiment in
  Gravitational Physics}}}\ (\bibinfo  {publisher} {Cambridge University
  Press},\ \bibinfo {year} {1993})\BibitemShut {NoStop}%
\bibitem [{\citenamefont {{Mendoza}}\ \emph {et~al.}(2013)\citenamefont
  {{Mendoza}}, \citenamefont {{Bernal}}, \citenamefont {{Hernandez}},
  \citenamefont {{Hidalgo}},\ and\ \citenamefont {{Torres}}}]{mendoza13}%
  \BibitemOpen
  \bibfield  {author} {\bibinfo {author} {\bibfnamefont {S.}~\bibnamefont
  {{Mendoza}}}, \bibinfo {author} {\bibfnamefont {T.}~\bibnamefont {{Bernal}}},
  \bibinfo {author} {\bibfnamefont {X.}~\bibnamefont {{Hernandez}}}, \bibinfo
  {author} {\bibfnamefont {J.~C.}\ \bibnamefont {{Hidalgo}}},\ and\ \bibinfo
  {author} {\bibfnamefont {L.~A.}\ \bibnamefont {{Torres}}},\ }\href
  {https://doi.org/10.1093/mnras/stt752} {\bibfield  {journal} {\bibinfo
  {journal} {Monthly Notices of the Royal Astronomical Society}\ }\textbf
  {\bibinfo {volume} {433}},\ \bibinfo {pages} {1802} (\bibinfo {year}
  {2013})},\ \Eprint {https://arxiv.org/abs/1208.6241} {arXiv:1208.6241
  [astro-ph.CO]} \BibitemShut {NoStop}%
\bibitem [{\citenamefont {{Mendoza}}\ and\ \citenamefont
  {{Olmo}}(2015)}]{mendozaolmo}%
  \BibitemOpen
  \bibfield  {author} {\bibinfo {author} {\bibfnamefont {S.}~\bibnamefont
  {{Mendoza}}}\ and\ \bibinfo {author} {\bibfnamefont {G.~J.}\ \bibnamefont
  {{Olmo}}},\ }\href {https://doi.org/10.1007/s10509-015-2363-y} {\bibfield
  {journal} {\bibinfo  {journal} {\apss}\ }\textbf {\bibinfo {volume} {357}},\
  \bibinfo {eid} {133} (\bibinfo {year} {2015})},\ \Eprint
  {https://arxiv.org/abs/1401.5104} {arXiv:1401.5104 [gr-qc]} \BibitemShut
  {NoStop}%
\bibitem [{\citenamefont {{Mendoza}}(2015)}]{mendoza15}%
  \BibitemOpen
  \bibfield  {author} {\bibinfo {author} {\bibfnamefont {S.}~\bibnamefont
  {{Mendoza}}},\ }\href {https://doi.org/10.1139/cjp-2014-0208} {\bibfield
  {journal} {\bibinfo  {journal} {Canadian Journal of Physics}\ }\textbf
  {\bibinfo {volume} {93}},\ \bibinfo {pages} {217} (\bibinfo {year}
  {2015})}\BibitemShut {NoStop}%
\bibitem [{\citenamefont {{Capozziello}}\ and\ \citenamefont {{De
  Laurentis}}(2012)}]{capozziello12}%
  \BibitemOpen
  \bibfield  {author} {\bibinfo {author} {\bibfnamefont {S.}~\bibnamefont
  {{Capozziello}}}\ and\ \bibinfo {author} {\bibfnamefont {M.}~\bibnamefont
  {{De Laurentis}}},\ }\href {https://doi.org/10.1002/andp.201200109}
  {\bibfield  {journal} {\bibinfo  {journal} {Annalen der Physik}\ }\textbf
  {\bibinfo {volume} {524}},\ \bibinfo {pages} {545} (\bibinfo {year}
  {2012})}\BibitemShut {NoStop}%
\bibitem [{\citenamefont {{Capozziello}}\ \emph {et~al.}(2007)\citenamefont
  {{Capozziello}}, \citenamefont {{Stabile}},\ and\ \citenamefont
  {{Troisi}}}]{capozziello07}%
  \BibitemOpen
  \bibfield  {author} {\bibinfo {author} {\bibfnamefont {S.}~\bibnamefont
  {{Capozziello}}}, \bibinfo {author} {\bibfnamefont {A.}~\bibnamefont
  {{Stabile}}},\ and\ \bibinfo {author} {\bibfnamefont {A.}~\bibnamefont
  {{Troisi}}},\ }\href {https://doi.org/10.1103/PhysRevD.76.104019} {\bibfield
  {journal} {\bibinfo  {journal} {\prd}\ }\textbf {\bibinfo {volume} {76}},\
  \bibinfo {eid} {104019} (\bibinfo {year} {2007})},\ \Eprint
  {https://arxiv.org/abs/0708.0723} {arXiv:0708.0723 [gr-qc]} \BibitemShut
  {NoStop}%
\bibitem [{\citenamefont {{Capozziello}}\ and\ \citenamefont {{de
  Laurentis}}(2011)}]{capozziello11}%
  \BibitemOpen
  \bibfield  {author} {\bibinfo {author} {\bibfnamefont {S.}~\bibnamefont
  {{Capozziello}}}\ and\ \bibinfo {author} {\bibfnamefont {M.}~\bibnamefont
  {{de Laurentis}}},\ }\href {https://doi.org/10.1016/j.physrep.2011.09.003}
  {\bibfield  {journal} {\bibinfo  {journal} {\physrep}\ }\textbf {\bibinfo
  {volume} {509}},\ \bibinfo {pages} {167} (\bibinfo {year} {2011})},\ \Eprint
  {https://arxiv.org/abs/1108.6266} {arXiv:1108.6266 [gr-qc]} \BibitemShut
  {NoStop}%
\bibitem [{\citenamefont {{De Martino}}\ \emph {et~al.}(2014)\citenamefont {{De
  Martino}}, \citenamefont {{De Laurentis}}, \citenamefont
  {{Atrio-Barandela}},\ and\ \citenamefont {{Capozziello}}}]{demartino14}%
  \BibitemOpen
  \bibfield  {author} {\bibinfo {author} {\bibfnamefont {I.}~\bibnamefont {{De
  Martino}}}, \bibinfo {author} {\bibfnamefont {M.}~\bibnamefont {{De
  Laurentis}}}, \bibinfo {author} {\bibfnamefont {F.}~\bibnamefont
  {{Atrio-Barandela}}},\ and\ \bibinfo {author} {\bibfnamefont
  {S.}~\bibnamefont {{Capozziello}}},\ }\href
  {https://doi.org/10.1093/mnras/stu903} {\bibfield  {journal} {\bibinfo
  {journal} {\mnras}\ }\textbf {\bibinfo {volume} {442}},\ \bibinfo {pages}
  {921} (\bibinfo {year} {2014})},\ \Eprint {https://arxiv.org/abs/1310.0693}
  {arXiv:1310.0693 [astro-ph.CO]} \BibitemShut {NoStop}%
\bibitem [{\citenamefont {{De Martino}}\ \emph {et~al.}(2018)\citenamefont {{De
  Martino}}, \citenamefont {{Lazkoz}},\ and\ \citenamefont {{De
  Laurentis}}}]{demartino18}%
  \BibitemOpen
  \bibfield  {author} {\bibinfo {author} {\bibfnamefont {I.}~\bibnamefont {{De
  Martino}}}, \bibinfo {author} {\bibfnamefont {R.}~\bibnamefont {{Lazkoz}}},\
  and\ \bibinfo {author} {\bibfnamefont {M.}~\bibnamefont {{De Laurentis}}},\
  }\href {https://doi.org/10.1103/PhysRevD.97.104067} {\bibfield  {journal}
  {\bibinfo  {journal} {\prd}\ }\textbf {\bibinfo {volume} {97}},\ \bibinfo
  {eid} {104067} (\bibinfo {year} {2018})},\ \Eprint
  {https://arxiv.org/abs/1801.08135} {arXiv:1801.08135 [gr-qc]} \BibitemShut
  {NoStop}%
\bibitem [{\citenamefont {{De Laurentis}}\ \emph {et~al.}(2018)\citenamefont
  {{De Laurentis}}, \citenamefont {{De Martino}},\ and\ \citenamefont
  {{Lazkoz}}}]{delaurentis18}%
  \BibitemOpen
  \bibfield  {author} {\bibinfo {author} {\bibfnamefont {M.}~\bibnamefont {{De
  Laurentis}}}, \bibinfo {author} {\bibfnamefont {I.}~\bibnamefont {{De
  Martino}}},\ and\ \bibinfo {author} {\bibfnamefont {R.}~\bibnamefont
  {{Lazkoz}}},\ }\href {https://doi.org/10.1103/PhysRevD.97.104068} {\bibfield
  {journal} {\bibinfo  {journal} {\prd}\ }\textbf {\bibinfo {volume} {97}},\
  \bibinfo {eid} {104068} (\bibinfo {year} {2018})},\ \Eprint
  {https://arxiv.org/abs/1801.08136} {arXiv:1801.08136 [gr-qc]} \BibitemShut
  {NoStop}%
\bibitem [{\citenamefont {{Cruz-Osorio}}\ \emph {et~al.}(2021)\citenamefont
  {{Cruz-Osorio}}, \citenamefont {{Gimeno-Soler}}, \citenamefont {{Font}},
  \citenamefont {{De Laurentis}},\ and\ \citenamefont
  {{Mendoza}}}]{cruz-osorio21}%
  \BibitemOpen
  \bibfield  {author} {\bibinfo {author} {\bibfnamefont {A.}~\bibnamefont
  {{Cruz-Osorio}}}, \bibinfo {author} {\bibfnamefont {S.}~\bibnamefont
  {{Gimeno-Soler}}}, \bibinfo {author} {\bibfnamefont {J.~A.}\ \bibnamefont
  {{Font}}}, \bibinfo {author} {\bibfnamefont {M.}~\bibnamefont {{De
  Laurentis}}},\ and\ \bibinfo {author} {\bibfnamefont {S.}~\bibnamefont
  {{Mendoza}}},\ }\href {https://doi.org/10.1103/PhysRevD.103.124009}
  {\bibfield  {journal} {\bibinfo  {journal} {\prd}\ }\textbf {\bibinfo
  {volume} {103}},\ \bibinfo {eid} {124009} (\bibinfo {year} {2021})},\ \Eprint
  {https://arxiv.org/abs/2102.10150} {arXiv:2102.10150 [astro-ph.HE]}
  \BibitemShut {NoStop}%
\end{thebibliography}%
